\documentclass[english,runningheads,a4paper]{llncs}

\usepackage{amsmath}
\usepackage{amssymb}
\usepackage{graphicx}
\usepackage{epstopdf}
\usepackage{algorithm}
\usepackage{algorithmic}
\usepackage{comment}
\usepackage{subfigure}

\usepackage{url}
\urldef{\mailsa}\path|{alfred.hofmann, ursula.barth, ingrid.haas, frank.holzwarth,|
\urldef{\mailsb}\path|anna.kramer, leonie.kunz, christine.reiss, nicole.sator,|
\urldef{\mailsc}\path|erika.siebert-cole, peter.strasser, lncs}@springer.com|
\newcommand{\keywords}[1]{\par\addvspace\baselineskip
\noindent\keywordname\enspace\ignorespaces#1}

\makeatletter

\newtheorem{thm}{\protect\theoremname}
\newtheorem{prop}[thm]{Proposition}
\newtheorem{lem}[thm]{Lemma}
\newtheorem{defin}[thm]{\protect\definname}
\newtheorem{observ}[thm]{\protect\observname}

\newtheorem{cnj}[thm]{Conjecture}

\makeatother

\usepackage{babel}

\providecommand{\theoremname}{Theorem}
\providecommand{\definname}{Definition}
\providecommand{\observname}{Observation}
\providecommand{\corolname}{Corollary}

\providecommand{\problemname}{Problem}
\usepackage{pgfplots}
\usetikzlibrary{calc}
\usepackage{pgfplotstable}
\usepackage{colortbl}
\usepackage{booktabs}
\usepackage{pgfplotstable}

\newcommand{\nix}[1]{{}}

\begin{document}

\mainmatter  % start of an individual contribution

\title{A Note on Optimality of Quantum Circuits over Metaplectic Basis}

\titlerunning{On Optimality of Quantum Circuits over Metaplectic Basis}

\author{Alex Bocharov$^{1}$}
\authorrunning{Alex Bocharov}


\institute{
$^1$Quantum Architectures and Computations Group, Microsoft Research, Redmond, WA (USA)\\
alexeib@microsoft.com\\
\url{http://research.microsoft.com/en-us/groups/quarc}}


\toctitle{}
%\tocauthor{Authors' Instructions}
\maketitle

\begin{abstract}
Metaplectic quantum basis is a universal multi-qutrit quantum basis, formed by the ternary Clifford group and the axial reflection gate
$R=|0\rangle \langle 0|  + |1\rangle \langle 1| - |2\rangle \langle 2|$. It is arguably, a ternary basis with the simplest
geometry.
 Recently Cui, Kliuchnikov, Wang and the Author have
proposed a compilation algorithm to approximate any two-level Householder reflection to precision $\varepsilon$ by a metaplectic circuit of $R$-count at most
$C \, \log_3(1/\varepsilon) + O(\log \log 1/\varepsilon)$ with $C=8$. A new result in this note takes the constant down to $C=5$ for non-exceptional target reflections under a certain credible number-theoretical conjecture. The new method increases the chances  of obtaining a truly optimal circuit but may not guarantee the true optimality.
Efficient approximations of an important ternary quantum gate proposed by Howard, Campbell and others is also discussed.  Apart from this, the note is mostly didactical: we demonstrate how to leverage Lenstra's integer geometry algorithm from 1983 for circuit synthesis.


\end{abstract}

\section{Introduction}

Metaplectic basis is one of the two simple universal quantum bases available on ternary (multi-qutrit) quantum computer (cf. \cite{BCRS},\cite{BMS}). In has been realized, in particular, by braiding and topological measurement within a framework of certain weakly-integral non-abelian anyons (\cite{CuiHongWang},\cite{CuiWang}). This makes the task of synthesizing efficient quantum circuits in the basis important in the context of topological quantum computation. In \cite{BCKW} we have developed effective methods for approximating arbitrary multi-qutrit unitaries by efficient metaplectic circuits. A certain probabilistic search algorithm has been at the core of that methodology.
 It was designed to $\varepsilon$-approximate a two-level unitary vector in $\mathbb{C}^{3^n}$ by a metaplectic circuit with $R$-count in
$4 \, \log_3(1/\varepsilon) + O(\log \log 1/\varepsilon)$.

It follows that given a two-level vector $u \in \mathbb{C}^{3^n}$ the \emph{Householder reflection}

\[
I^{\otimes n} - 2 |u\rangle \langle u|
\]
\noindent can be so approximated at the $R$-count in $8 \, \log_3(1/\varepsilon) + O(\log \log 1/\varepsilon)$.

In this note we propose a deterministic-search algorithm that looks for a truly optimal approximation of a two-level unitary vector. Let $v$ be such a vector and $\varepsilon>0$ be an arbitrarily small target precision. If $v$ is in general position with respect of Eisenstein lattice introduced in section \ref{subsec:eisen:lattice}, then (under a certain plausible number-theory conjecture) the algorithm finds an $\varepsilon$-approximation for $v$ as a metaplectic circuit with $R$-count at most ${\bf 5/2} \, \log_3(1/\varepsilon) + O(\log \log 1/\varepsilon)$. The approximation found in this class has the truly optimal depth if we have a certain kind of oracle to effectively solve any \emph{norm equation} over the ring of cyclotomic integers of degree $3$. In practice we can build such an oracle that would be as efficient as any available procedure for factorization of rational integers. Unless $(v,\varepsilon)$ forms an \emph{exceptional pair} as defined in section \ref{subsec:unduc:except}, the algorithm terminates in classical time that is polylogarithmic in $1/\varepsilon$.

The evolution from the probabilistic search to deterministic search has become typical, recently, in the multi-qubit circuit synthesis (see, for example \cite{RoSelinger}, \cite{Ross}, \cite{BBG}). It may be  strongly related to general lattice reduction methods, as suggested in \cite{KBRY}.

The note is organized as follows: the preliminaries on the Clifford+$R$ basis are given in section \ref{sec:prelim}; this is followed by the description and analysis of the synthesis method in section \ref{sec:main:lemma}; section \ref{sec:P9} describes a special case method for approximating the important ternary $P_9$ gate that had been introduced in \cite{CampbellEtAl}, \cite{Howard}  and recently used to build efficient ternary arithmetics circuits in \cite{BCRS}.

\section{Preliminaries} \label{sec:prelim}
In this section we focus on facts about single-qutrit Clifford group and its extension into a universal metaplectic basis.

Let $\{|0\rangle,|1\rangle,|2\rangle\}$ be the standard computational basis for a qutrit.
Let $\omega_3 = e^{2 \pi \, i/3}$ be the third primitive root of unity.
The ternary \emph{Pauli} group is generated by the \emph{increment} gate
\begin{equation} \label{eq:INC}
\mbox{INC}= |1\rangle \langle 0|+|2\rangle \langle 1|+|0\rangle \langle 2|
\end{equation}
\noindent and the ternary $Z$ gate
\begin{equation} \label{eq:Z}
Z= |0\rangle \langle 0|+\omega_3 |1\rangle \langle 1|+\omega_3^2 |2\rangle \langle 2|.
\end{equation}

The single-qutrit ternary \emph{Clifford} group stabilizes the Pauli group and is generated by the ternary Hadamard gate $H$,
\begin{equation} \label{eq:H}
H = \frac{1}{\sqrt{3}} \sum \omega_3^{j\, k} |j\rangle \langle k|,
\end{equation}
\noindent and the $Q$ gate
\begin{equation} \label{eq:Q}
Q= |0\rangle \langle 0|+|1\rangle \langle 1|+\omega_3 |2\rangle \langle 2|,
\end{equation}

Compared to the binary Clifford group, $H$ is the ternary counterpart of the binary Hadamard gate, $Q$ is the counterpart of the phase gate $S$. It is seen, in particular,  that $H^2$ and $\mbox{INC}$ generate all 6 classical permutations of the standard qutrit basis.

The ternary Clifford group is finite. It can be extended to the quantum universal \emph{metaplectic} basis (see \cite{CuiWang}, \cite{BCKW}) by adding the non-Clifford axial reflection gate
\begin{equation}
R=|0\rangle \langle 0|  + |1\rangle \langle 1| - |2\rangle \langle 2|
\end{equation}

A circuit over the metaplectic basis is called a \emph{metaplectic circuit} and the number of occurrences of the $R$ gate in such circuit is called the \emph{$R$-count} of the circuit.

We recall that since as early as \cite{BarencoEtAl1995}, \cite{IkeAndMike2000}, it has been customary to synthesize ancilla-free circuits for single-qubit unitaries by splitting the target unitary into axial rotations first.
In \cite{BCKW} we have argued that in the single-qutrit case a more appropriate elementary building block is a two-level Householder reflection operator.

Given a unitary vector $u$ with at most two non-zero components, for example $u = x \, |0\rangle + y \, |1\rangle, \, x,y \in \mathbb{C}$ the corresponding Householder reflection is

\[
R_u = I - 2 |u\rangle \langle u|
\]

Any unitary from $\mbox{SU}(3)$ can be decomposed exactly into a product of at most 6 two-level Householder reflections at negligible classical cost. A core result in \cite{BCKW} is that a two-level Householder reflection can be $\varepsilon$-approximated by a metaplectic
circuit with the $R$-count at most $8 \, \log_3(1/\varepsilon) + O(\log \log 1/\varepsilon)$.

This has been derived from the fact that a two-level unitary state $u = x \, |0\rangle + y \, |1\rangle$ can be prepared at precision $\Theta(\varepsilon)$ by a metaplectic
circuit with the $R$-count at most $4 \, \log_3(1/\varepsilon) + O(\log \log 1/\varepsilon)$.

It has been conjectured in \cite{BCKW} that perhaps the latter state preparation can be done by a circuit with the $R$-count at most $5/2 \, \log_3(1/\varepsilon) + O(\log \log 1/\varepsilon)$ and that such circuit can be compiled in classical time that is polynomial in
$\log(1/\varepsilon)$. One of the goals of this note is to prove that conjecture for non-exceptional two-level unitary states.

\section{The main technical lemma} \label{sec:main:lemma}

\subsection{The meniscus} \label{subsec:meniscus}

Let $x = x_0\, |0\rangle + x_1 \, |1\rangle$  be a two-level unitary vector in the qutrit state space $\mathbb{C}^3$.

We start with a simple observation that the distance from $x$ to another unitary vector $y = y_0\, |0\rangle + y_1 \, |1\rangle + y_2 \, |2\rangle$ does not depend on $y_2$.

Indeed:
\begin{equation} \label{eq:distance:2}
|x-y|^2 = \langle x-y|x-y\rangle = |x|^2 - \langle x|y\rangle - \langle y|x\rangle + |y|^2 = 2 (1 - \mbox{Re}(\langle x|y\rangle))
\end{equation}
\noindent and $\langle x,y\rangle = x_0 \, y_0^* + x_1 \, y_1^*$ does not depend on $y_2$.

Let us regard the subspace $\mathbb{C}^2$ spanned by $|0\rangle$ and $|1\rangle$ as a four-dimensional real space with natural coordinates corresponding to
$\mbox{Re}(x_0), \mbox{Im}(x_0),\mbox{Re}(x_1),\mbox{Im}(x_1)$.

Rewriting (\ref{eq:distance:2}) a bit further we find

$\mbox{Re}(\langle x|y\rangle) = \mbox{Re}(x_0)\,\mbox{Re}(y_0) + \mbox{Im}(x_0)\,\mbox{Im}(y_0)+\mbox{Re}(x_1)\,\mbox{Re}(y_1)+\mbox{Im}(x_1)\, \mbox{Im}(y_1)$

\noindent which is precisely the length of the Euclidean projection of the vector $r[y] = (\mbox{Re}(y_0), \mbox{Im}(y_0),\mbox{Re}(y_1),\mbox{Im}(y_1))$ onto the unit vector

$r[x]=(\mbox{Re}(x_0), \mbox{Im}(x_0),\mbox{Re}(x_1),\mbox{Im}(x_1))$.

Note that $|x-y|<\varepsilon$ iff $\langle r[x],r[y]\rangle > 1-\varepsilon^2/2$. Thus we obtain a Euclidean geometric interpretation of the $\varepsilon$-closeness of a vector to the two-level unitary vector $x$ in the form of the following
\begin{observ}
A unitary vector $y$ is $\varepsilon$-close to the two-level unitary vector $x$ if and only if $\mbox{Re}(\langle x|y\rangle) > 1-\varepsilon^2/2$, if and only if the length of the Euclidean projection of the real vector $r[y]$ onto $r[x]$ is greater than $1-\varepsilon^2/2$.
\end{observ}

Given a real four-dimensional unit vector $p\in \mathbb{R}^4, |p|=1$ and a small enough $\varepsilon>0$ we define the following convex body
\begin{defin} \label{defin:meniscus}
$M_{\varepsilon}(p) = \{ q \in  \mathbb{R}^4 | |q|\leq 1, \langle q,p\rangle > 1 - \varepsilon^2/2 \}$ is called the $\varepsilon$-\emph{meniscus} around the vector $p$.
\end{defin}

The importance of the meniscus $M_{\varepsilon}(r[x])$ is that it contains $r[y]$ for any subunitary $y$ that is $\varepsilon$-close to $x$. In the next subsection we reinterprect the approximation problem as the problem of enumeration of certain lattice contained in a scaled image of the meniscus $M_{\varepsilon}(r[x])$.

\subsection{The Eisenstein lattice} \label{subsec:eisen:lattice}
Let $\omega_3=e^{2 \pi \, i /3}$ denote the third root of unity and let $\mathbb{Z}[\omega_3]$ be the ring of cyclotomic integers of degree 3.

It has been shown in \cite{BCKW} that a single-qutrit unitary state $|s\rangle$ can be prepared exactly from a standard basis state (say, $|2\rangle$) by a finite metaplectic circuit if and only if the state has the form
\begin{equation} \label{eq:exact:state}
|s\rangle = (u\, |0\rangle + v\, |1\rangle + w\, |2\rangle )/\sqrt{3}^k, u,v,w \in \mathbb{Z}[\omega_3], k \in \mathbb{Z}, |u|^2+|v|^2+|w|^2 = 3^k
\end{equation}

Assuming $k$ is minimal for all equivalent representations of $|s\rangle$ of this form, then $|s\rangle$ can be prepared by a metaplectic circuit of overall depth between $2\,k-1$ and $2\,k+1$ with $R$-count at most $k+1$.

Given a state of form (\ref{eq:exact:state}) we are going to take its first two coordinates $u/\sqrt{3}^k$, $v/\sqrt{3}^k$ only and rewrite them as the four-dimensional real vector in $\mathbb{R}^4$.
Recall that $\omega_3=-1/2 + \sqrt{3}/2 \, i$ and that therefore a cyclotomic integer $u$ is of the form $a+b \, \omega_3 = (a-b/2)+b \, \sqrt{3}/2 \, i, a, b \in \mathbb{Z}$.

Therefore $r[|s\rangle]$ is a linear combination with integer coefficients of the following four basis vectors:

\begin{equation} \label{eq:v:basis}
v_1 = (1/\sqrt{3}^k, 0,0,0), v_2 =(-1/(2 \,\sqrt{3}^k), \sqrt{3}/(2 \,\sqrt{3}^k),0,0),
\end{equation}
\begin{equation} \label{eq:v:basis:contd}
v_3 =(0,0,1/\sqrt{3}^k,0),v_4=(0,0,-1/(2 \,\sqrt{3}^k), \sqrt{3}/(2 \,\sqrt{3}^k))
\end{equation}

\noindent Thus for any fixed $k$ the set of all such $r[|s\rangle]$ forms a four-dimensional integer lattice with the basis vectors shown.
In this note we call this structure the \emph{scaled Eisenstein integer lattice} with the scaling coefficient of $1/\sqrt{3}^k$ and denote it by the $\mathcal{E}_k$.
Let now $x = x_0\, |0\rangle + x_1 \, |1\rangle$ be the target two-level unitary state and suppose we have found a point $a_1 \, v_1 + a_2 \, v_2 +a_3 \, v_3 +a_4 \, v_4$ in the lattice $\mathcal{E}_k$ that also belongs to the meniscus
$M_{\varepsilon}(r[x])$ (as defined in the previous subsection). Set $u = (a_1+a_2 \, \omega_3)/\sqrt{3}^k$, $v = (a_3+a_4 \, \omega_3)/\sqrt{3}^k$. This candidate solution can be bootstrapped into a state of the form (\ref{eq:exact:state}) that is $\varepsilon$-close to $x$ if and only if there existed a cyclotomic integer $w$ such that

\begin{equation} \label{eq:norm:eq}
|w|^2 = 3^k - |u|^2 -|v|^2
\end{equation}

The converse is trivially true. Any $\varepsilon$-approximation of $x$ of the form (\ref{eq:exact:state}) generates a unique point in $\mathcal{E}_k$  and $M_{\varepsilon}(r[x])$ along with a solution $w$ of the equation (\ref{eq:norm:eq}).

We recall from \cite{BCKW} that an optimal metaplectic circuit preparing the state (\ref{eq:exact:state}) from a standard basis state has the $R$-count at most $k+1$ and overall depth in $\{2\,k-1,2\,k,2\,k+1\}$. In the above approximation procedure, if $k$ is absolutely minimal then the depth of the corresponding circuit cannot exceed the absolute minimum by more than 2. Thus the task of building a near-optimal $\varepsilon$-approximation circuit falls, characteristically, into two parts: 1) enumerating all the points of the lattice $\mathcal{E}_k$  within the meniscus $M_{\varepsilon}(r[x])$ and 2) solving the norm equation (\ref{eq:norm:eq}) for those points for which it is solvable.

 Detailed analysis of the second part of the task would be largely outside the scope of this paper. (The reader is encouraged to confer with the Appendix B of \cite{BCKW}). For completeness we briefly discuss the key conjecture behind that second part.

  Depending on the right hand side, an instance of the equation (\ref{eq:norm:eq}) can be unsolvable or easy to solve classically, or hard to solve classically. We have shown that limiting the circuit compilation to easily solvable instances only, still allows to achieve asymptotically optimal depth of the resulting circuits; however to find circuits of near-optimal depth we need to also address the hard cases (since there is always a chance that the true optimum corresponds to a hard norm equation).

For completeness we summarize an important number theory conjecture from the \cite{BCKW},Appendix A:
\begin{cnj} \label{cnj:norm:eq}
Let $k$ be arbitrarily large positive integer and let $u,v\in \mathbb{Z}[\omega_3]$ be randomly picked Eisenstein integers such that
\[
\Theta(3^{k/2}) \leq |u|^2+|v|^2 \leq 3^k
\]
Then the equation (\ref{eq:norm:eq}) is solvable for $w \in \mathbb{Z}[\omega_3]$ with probability that has uniform lower bound in $\Omega(1/k)$.
\end{cnj}

Due to recent progress in strong approximation for quadratic forms as described in \cite{Sardari}, this conjecture has become almost irrelevant for approximation techniques comparable to  \cite{BCKW}, Lemma 12. Indeed, it follows from Theorem 1.8 of \cite{Sardari} that for arbitrarily small preselected $\delta>0$, for any single-qutrit unitary vector $x$ and for arbitrarily small desired precision $\varepsilon>0$ there exists a unitary vector of the form (\ref{eq:exact:state}) within the distance $\varepsilon$ of $x$ with $k \leq (4+\delta) \log_3(1/\varepsilon) + K(\delta)$, where $K(\delta)$ is some constant depending only on $\delta$. This conclusion is not based on the conjecture and for practical purposes it is perhaps just as good as the $k \leq C\, \log_3(1/\varepsilon) + O(\log \log 1/\varepsilon), \, C=4$ bound of \cite{BCKW}.

However, at the moment the Author is not aware of any stronger result that would similarly allow to disregard Conjecture \ref{cnj:norm:eq} when one wants to establish a bound on $k$ with $C<4$ for targets $x$ that are in general position.

Starting with \cite{RoSelinger} it is customary to talk about \emph{factorization oracle}, or, more generally, about \emph{norm equation oracle} needed for the hard cases of the norm equation. If, for example, we had a quantum computer capable of running Shor's factorization algorithm \cite{Shor} then we could solve all the solvable norm equations in polynomial time; otherwise we could still go for the hard cases on a classical computer by allowing a certain failure rate caused by runtime cutoffs.

Either way, if we are able to filter out and handle solvable norm equations, then it should take polynomially many (in $k$) candidate equations with different right hand sides to find a one that can be solved in polynomial time.

\subsection{Induction in $k$ and exceptions.} \label{subsec:unduc:except}

Let $k$ be the scaling index of the scaled Eisenstein lattice $\mathcal{E}_k$.
It is convenient to apply the scaling factor of $\sqrt{3}^k$ to the meniscus $M_{\varepsilon}(p)$ described in section \ref{subsec:meniscus}.

Consider the scaled-out meniscus $M=\sqrt{3}^k \, M_{\varepsilon}(p)$ and the unscaled Eisenstein lattice $\mathcal{E}$ generated by $(1,0),(\omega_3,0),(0,1),(0,\omega_3)$. If we wanted to enumerate points in $\mathcal{E}_k \cap M_{\varepsilon}(p)$, this would
immediately translate into enumeration of the points in $\mathcal{E} \cap (\sqrt{3}^k \, M_{\varepsilon}(p))$.

\begin{defin} \label{def:k:feasible}
Let $q \in \mathcal{E} \cap \sqrt{3}^k \, M_{\varepsilon}(p)$ and cast $q$ as the pair $(u,v) \in  \mathbb{Z}[\omega_3] \times \mathbb{Z}[\omega_3]$.
The point $q$ is called \emph{$k$-feasible} if the $|w|^2 = 3^k - |u|^2 - |v|^2$ is solvable for $w \in \mathbb{Z}[\omega_3]$.
\end{defin}

An informal top-level description of an algorithm for finding an Eisenstsein approximation of the given two-level unitary state with the optimal $k$ is as follows:

iterate through $k=0,1,2,\ldots$

for every $k$ iterate through $\mathcal{E} \cap \sqrt{3}^k \, M_{\varepsilon}(p)$ until the set is exhausted or a $k$-feasible point is found; in the latter case the current $k$ is the minimal feasible one.

It is easy to conclude from the Conjecture \ref{cnj:norm:eq} that, under the Conjecture the probability that a feasible lattice point had not yet been found decays exponentially with the number of the lattice points encountered and inspected. Lemma \ref{lem:fritfully: multiply} implies that if $\mathcal{E} \cap \sqrt{3}^{k_0} \, M_{\varepsilon}(p)$ contains at least two distinct points for some $k_0$ the search for a feasible point is going to terminate with probability arbitrarily close to 1 at some $k=k_0+O(\log(k_0))$.

Before moving forward let us prove the following
\begin{lem} \label{lem:fritfully: multiply}
If the scaled meniscus $\sqrt{3}^{k_0} \, M_{\varepsilon}(p)$ contains two distinct points $y_1,y_2$ of the Eisenstein lattice, then for any integer $\ell \leq 0$ the scaled meniscus $M_{\ell}=\sqrt{3}^{k_0+2 \ell} \, M_{\varepsilon}(p)$ contains at
least $3^{\ell}+1$ points of the lattice.
\end{lem}
\begin{proof}
This easily follows from the convexity of the meniscuses. The segment $[y_1,y_2]$ belongs to $M_{0}$. We note
%that $3^{\ell}\,y_1,3^{\ell}\,y_2$ are lattice points and
that $r\,y_1+(3^{\ell}-r)\,y_2, \, r=0,\ldots,3^{\ell}$ are distinct lattice points in $M_{\ell}$. (Each of them is homothetic with coefficient $3^{\ell}$ to some point in $[y_1,y_2]$).
\end{proof}

Thus successful termination is ensured (with probability arbitrarily close to 1) as soon as we have found such $k_0$. The latter however runs into some complications analyzed below.

 We start with the following volume argument: given a $k$ in $5/2 \, \log_3(1/\varepsilon) + O(\log \log 1/\varepsilon)$, then the Euclidean 4-volume of the $M=\sqrt{3}^k \, M_{\varepsilon}(p)$ is in
 
$O(\log(1/\varepsilon)^{\nu})$ where $\nu$ is some constant derived from the leading coefficient of the $O(\log \log 1/\varepsilon)$ term. This is easily established by direct algebraic manipulation.

Thus the 4-volume of the the scaled meniscus is polylogarithmic in $1/\varepsilon$ when $k$ follows the above asymptotics. It would be tempting to say that for such $k$ the number of the Eisenstein lattice points in the scaled meniscus is in $\Theta(\log(1/\varepsilon)^{\nu})$ based on the so called Gaussian heuristics (cf. \cite{LLL}, Chapter 2.). Unfortunately the geometry here is significantly more complicated and the Gauss heuristics is known to fail in such cases.  Indeed for such $k$ the width of the scaled meniscus $M$ in the direction of $p$ is in $O(\varepsilon^{3/4} \, \log(1/\varepsilon)^{\nu/4})$ and thus $M$ is asymptotically flat in that direction. It is conceivable that $M$ may contain no Eisenstein lattice points at all or it may contain roughly
$(1/\varepsilon)^{3/4}$ lattice points (in case it contains a segment of a rank $3$ sublattice of $\mathbb{E}$)\footnote{The Author is indebted to Vadym Kliuchnikov for pointing out this complication.}.

%We note, however that, because the scaled meniscus is asymptotically flat in only one direction, the $O((1/\varepsilon)^{3/4})$ is the worst case asymptotic upper bound on the number of lattice points in the scaled meniscus.

In what follows we reserve the symbol $\lambda$ to denote some positive real constant $0 < \lambda \leq 3$.
\begin{comment}
We say that a triplet
$(p,k,\varepsilon), \, p\in \mathbb{R}^4, k\in \mathbb{Z}, \varepsilon>0$ is \emph{exceptional at level $\lambda$}
if the scaled meniscus $\sqrt{3}^k M_{\varepsilon}(p)$ contains at least $\varepsilon^{-\lambda}$ points.

In follows from the Lemma \ref{lem:fritfully: multiply} that

(a) given an exceptional triplet  $(p,k,\varepsilon)$ at level $\lambda$ and an integer $\ell \leq 0$, then the triplet $(p, k+2\, \ell,\varepsilon)$ is exceptional at a level greater than $\lambda$, and

(b) if  $\sqrt{3}^{k_0} \, M_{\varepsilon}(p)$ contains at least two distinct points, then for a fixed $\lambda$ the triplet
$(p,k,\varepsilon)$ becomes exceptional for a large enough $k\geq k_0$. This motivates the following
\end{comment}

\begin{defin}
A pair $(p \in \mathbb{R}^4,\varepsilon>0)$ is called \emph{strongly exceptional at level} $\lambda>0$ if there exists a $k\in \mathbb{Z}$ such that the scaled meniscus $\sqrt{3}^k \, M_{\varepsilon}(p)$ contains at least $\varepsilon^{-\lambda}$ Eisenstein lattice points, but $\sqrt{3}^{\ell} \, M_{\varepsilon}(p)$ contains fewer than two lattice points for any $\ell < k$.
\end{defin}

Obviously strongly exceptional pairs form an open subset in $\mathbb{R}P^3 \times \mathbb{R}$ as small enough deformations of either $p$ or $\varepsilon$ preserve the exceptionality.

Informally, given a strongly exceptional pair $(p,\varepsilon)$ and as $k$ iterates from $0$ to near
$ 4\, \log_3(1/\varepsilon)$, the search sets $\mathcal{E} \cap \sqrt{3}^k \, M_{\varepsilon}(p)$ stay empty for a long time, and then we get ``exponentially'' many lattice points all at once. It is important to note that for $k > 4\, \log_3(1/\varepsilon)$ the width of the scaled meniscus in the direction $p$ is greater than 1, and its width in any orthogonal direction is asymptotically large in $1/\varepsilon$; therefore for such $k$ the $\mathcal{E} \cap \sqrt{3}^k \, M_{\varepsilon}(p)$ is non-empty and some of its points can be found effectively by simple roundoff.

Since our goal is to find lattice approximations with $k < 4\, \log_3(1/\varepsilon)$ then any exceptional pair is an obstacle to reaching this goal in two distinct ways. First, the desired lattice points likely do not exist. Secondly, even if they do exist for some $k_1 < 4\, \log_3(1/\varepsilon)$, the set $\mathcal{E} \cap \sqrt{3}^{k_1} \, M_{\varepsilon}(p)$ may be too large to be enumerated in a reasonable time. However, when the size of this set is $>> \, k_1$, it contains a $k_1$-feasible point with probability effectively 1, and such point can be found in $O(k_1)$ iterations. Thus the second aspect of the problem only matters if we are looking for truly optimal circuit depth or truly optimal $R$-count and must enumerate all the candidates.

For better intuition consider the following important
\begin{example}
Let $u$ be a vector of the dual Eisenstein lattice $\mathcal{E}^*$, i.e. such that
$\forall v \in \mathcal{E}, \langle u,v \rangle \in \mathbb{Z}$ and let us make a technical assumption that $|u|$ and $|u|\, \sqrt{3}$ are irrational. Then there exists a $\varepsilon_1>0$ such that for any $\varepsilon < \varepsilon_1$ the $(p=u/|u|,\varepsilon)$ is a strongly exceptional pair for some $\lambda > 1$.
\end{example}

A rigorous proof of this fact is beyond the scope of this note. Here is a qualitative explanation why this is the case. As per a known property of dual lattices, let $H$ be the hyperplane through the origin orthogonal to $p$; then any vector of the original lattice belongs to a shifted hyperplane $H_{\ell}=\ell/|u|\, p + H$ for some $\ell \in \mathbb{Z}$. Suppose $\ell = \lfloor \sqrt{3}^k \, |u| \rfloor$; then $H_{\ell}$ is the closest to point $\sqrt{3}^k \, p$ among such $H_{m}$ that intersect the ball of radius $\sqrt{3}^k$ centered at origin. More specifically, if $f_k = \sqrt{3}^k \, |u| -\ell$ then $\sqrt{3}^k \, p$ is at the distance $f_k/|u|$ from $H_{\ell}$. Clearly we are not getting any lattice points in the scaled meniscus $\sqrt{3}^k \, M_{\varepsilon}(p)$ unless $\sqrt{3}^k \varepsilon^2/2 > f_k/|u|$ or, equivalently 
\begin{equation} \label{eq:k:fk}
k > 4 \, \log_3(1/\varepsilon) - 2\,(\log_3(|u|/2)- \log_3(f_k))
\end{equation}
\noindent The $-2\, \log_3(|u|)$ is asymptotically irrelevant when $\varepsilon \rightarrow 0$. The irrationality assumption makes the inequality (\ref{eq:k:fk}) meaningful. For example, we can fix a large enough $k_0 \in \mathbb{Z}$ and consider only $\varepsilon < \mbox{min} \{ 3^{-(|u|/(2 \, f_k))^2} | k = 0, \ldots k_0\}$. Then for $k < \mbox{min}(k_0, 3 \, \log_3(1/\varepsilon))$ the scaled meniscus does not contain any Eisenstein lattice vectors. As $k$ grows towards $4 \, \log_3(1/\varepsilon)$ the first non-empty intersection $\mathcal{E} \cap \sqrt{3}^k \, M_{\varepsilon}(p)$ will get on the order of $\Theta(3^{3\,k/2})$ lattice points.

Suppose $\lambda$ is close to zero (e.g. take $\lambda=0.1$ for practical purposes). We will use it as a hyperparameter in the meta-algorithm \ref{alg:iteration} below. Our intention is to have this algorithm run in classical time that is polylogarithmic in $1/\varepsilon$ in most cases, and in subexponential time $O(\varepsilon^{-\lambda})$ in the remaining non-exceptional cases.

\begin{algorithm}[H]
\caption{Iterative search for a feasible lattice point.}
\label{alg:iteration}
\algsetup{indent=2em}
\begin{algorithmic}[1]
\REQUIRE{$p \in \mathbb{R}^4$, $\varepsilon>0$, hyperparameter $0 < \lambda \leq 3/4$}
\FOR {$k=0,1,\ldots$}
\STATE {$N \gets \#|\mathcal{E} \cap \sqrt{3}^k \, M_{\varepsilon}(p)|$}
\IF {$N > \varepsilon^{-\lambda}$}
\RETURN {``Exception''} \label{il1}
\ENDIF
\FOR {$\forall$ lattice point $q \in \mathcal{E} \cap \sqrt{3}^k \, M_{\varepsilon}(p)$}
\IF {$q$ ia a $k$-feasible point}
\RETURN {$(q,k)$} \label{il2}
\ENDIF
\ENDFOR
\ENDFOR
\end{algorithmic}
\end{algorithm}

This meta-algorithm terminates with either a feasible lattice point or with an ``Exception''. Assuming Conjecture \ref{cnj:norm:eq}, the probability that it terminates with  the ``Exception'' is near zero, unless the input pair $(p,\varepsilon)$ is strongly exceptional at the level $\lambda$.

Let us make this statement rigorous and also find an asymptotical bound on the number of iterations needed for termination.

Suppose $(p,\varepsilon)$ is not strongly exceptional. Find the smallest $k_1 \in \mathbb{Z}$ such that
$1 < \#|\mathcal{E} \cap \sqrt{3}^{k_1} \, M_{\varepsilon}(p)| < \varepsilon^{-\lambda}$.
For any $k<k_1$ there can be at most one lattice point candidate to inspect. If some such point $q$ is $k$-feasible then the search would succeed with $(q,k)$.

Suppose this hasn't happened. Let us iterate through integers $\ell=0,1,\ldots$. As per the Lemma \ref{lem:fritfully: multiply}, the $\#|\mathcal{E} \cap \sqrt{3}^{k_1+2 \ell} \, M_{\varepsilon}(p)|$ grows as $\Omega(3^{\ell})$ while as per the Conjecture \ref{cnj:norm:eq} our lower bound for the probability of a lattice point to be feasible does not change significantly with $\ell$ and stays in
$\Omega(1/\log(1/\varepsilon))$. Therefore the probability of not finding a feasible lattice point in $\mathcal{E} \cap \sqrt{3}^{k_1+2 \ell} \, M_{\varepsilon}(p)$ decays exponentially with $\ell$, and, conversely the probability of finding one can be made within an arbitrarily small $\delta>0$ from 1 for some $\ell$ in $\log \log 1/\varepsilon$. Thus, with probability arbitrarily close to 1 the search should succeed for some $k_2=k_1+O(\log \log 1/\varepsilon)$ after inspecting $O(\log 1/\varepsilon)$ lattice points.

\subsection{The two-level state approximation}

The purpose of this subsection is to offer an effective proof of the following

\begin{lem} \label{lem:state:approx:complex}
Let  $|\psi\rangle = x\, |0\rangle + y \, |1\rangle, x,y \in \mathbb{C}, |x|^2+|y|^2 =1$ be a two-level unitary vector and $\varepsilon>0$  be a sufficiently small precision value.

$|\psi\rangle$ can be approximated by a unitary vector of the form
\begin{equation} \label{eq:approximant}
(u\, |0\rangle + v \, |1\rangle + w\, |2\rangle )/\sqrt{3}^k, u,v,w \in \mathbb{Z}[\omega_3]
 \end{equation}
 \noindent with integer $k$ in  $C \, \log_3(1/\varepsilon) + O(\log \log 1/\varepsilon)$ where $C$ is some positive constant.

 For a pair $(|\psi\rangle, \varepsilon)$ in general position, the constant $C$ can be constrained to $C\leq {\bf 5/2}$.

 An approximating vector of the form (\ref{eq:approximant}) can be found by a classical search algorithm with runtime that is polylogarithmic in $1/\varepsilon$.

If an oracle for solving norm equations over the ring of cyclotomic integers of degree $3$ is available, then such a vector can be found with the absolutely minimal $k$ at the classical cost that is a polylogarithm in $1/\varepsilon$ times the classical cost of executing the oracle.
\end{lem}

This is a technical improvement on the Lemma 12 from \cite{BCKW}.

The main geometric tool for proving Lemma \ref{lem:state:approx:complex} is the following
\begin{lem} \label{lem:approx:subunitary}
Let $\varepsilon>0$ be an arbitrarily small precision and let $|\psi\rangle = x\, |0\rangle + y \, |1\rangle, x,y \in \mathbb{C}, |x|^2+|y|^2 =1$ be a two-level unitary vector in a qutrit.

 For a fixed integer $k$ the subunitary vectors of the form $(u\, |0\rangle + v \, |1\rangle)/\sqrt{3}^k, u,v \in \mathbb{Z}[\omega_3]$ within distance $\varepsilon$ from $|\psi\rangle$ can be enumerated effectively and deterministically. For a pair $(|\psi\rangle, \varepsilon)$ in general position such enumeration can be achieved in classical time that is polylogarithmic in $1/\varepsilon$ for any $k$ in $5/2 \, \log_3(1/\varepsilon) + O(\log \log 1/\varepsilon)$.
\end{lem}

As per our discussion in the subsection \ref{subsec:eisen:lattice}, this lemma can be translated into the lattice language as follows:

\begin{lem} \label{lem:enum:lattice}
In the context of Lemma \ref{lem:approx:subunitary} let $p=r[|\psi\rangle] \in \mathbb{R}^4$ be the corresponding 4-vector.

Given a $k$ in $5/2 \, \log_3(1/\varepsilon) + O(\log \log 1/\varepsilon)$, the points of the scaled Eisenstein lattice $\mathcal{E}_k$ within the meniscus $M_{\varepsilon}(p)$ can be enumerated effectively and deterministically in classical time that is polylogarithmic in $1/\varepsilon$.
\end{lem}

It is convenient to apply the scaling factor of $\sqrt{3}^k$ to the entire geometry of the Lemma, consider the scaled-out meniscus $M=\sqrt{3}^k \, M_{\varepsilon}(p)$ and the unscaled Eisenstein lattice $\mathcal{E}$ generated by $(1,0),(\omega_3,0),(0,1),(0,\omega_3)$. The enumeration task in the Lemma then
immediately translates into enumeration of the points in $\mathcal{E} \cap (\sqrt{3}^k \, M_{\varepsilon}(p))$.

We observe first that for the $k$ from Lemma \ref{lem:enum:lattice} the Euclidean 4-volume of the $M=\sqrt{3}^k \, M_{\varepsilon}(p)$ is polynomial in  $\log(1/\varepsilon)$ that therefore the number of the lattice points is upper-bounded by $\log(1/\varepsilon)^{\nu}$ where $\nu$ depends on the actual constants in the
$O(\log \log 1/\varepsilon)$ asymptotics term. The task is to be able to iterate through these points (if any) in time, polynomial in $\log(1/\varepsilon)$.

Our solution for this task takes some leveraging of the existing integer geometry tools but surprisingly little effort otherwise.

Let us look at the Euclidean dimensions of the scaled-out meniscus. Its width in the direction of $p$ is in $O(\varepsilon^{3/4}\, \log(1/\varepsilon)^{\nu/4})$ and thus asymptotically small, while its width in any orthogonal direction is in $O(\varepsilon^{-1/4}\, \log(1/\varepsilon)^{\nu/4})$ and thus asymptotically large.

A simple  observation is that we can enclose $M$ into a 4-dimensional polytope such that the Euclidean size of the polytope is some constant factor times the size of $M$.

If necessary we can also ensure that the vertices of such polytope are encoded by rational numbers and that  the bit sizes of these rational numbers are in $O(\log 1/\varepsilon)$.
For example, we can initially enclose $M$ into a 4-dimensional prism with real vertices. The 4-volume of such prism can be made less than 5 times larger than the 4-volume of $M$.

When rational  polytope is required , we would need to potentially round out the vertices of the prism to make them rational vectors. Because $M$ is only roughly $O(\varepsilon^{3/4})$ wide in one direction it would be best to keep the rounding precision at around $\varepsilon$ in order to not blow up the width, but in any case this can be done with the denominators of the rational numbers involved well within the magnitude in $O(1/\varepsilon)$ and thus their bit sizes logarithmic in $1/\varepsilon$.

To get better transparency in what follows let us tweak the problem setting one more time. Consider the following linear automorphism

$\iota : \mathbb{R}^4 \rightarrow \mathbb{R}^4$, $\iota(e_1)=e_1,\, \iota(e_2)=-1/2 \, e_1 + \sqrt{3}/2 \, e_2, \iota(e_3)=e_3,\, \iota(e_4)=-1/2 \, e_3 + \sqrt{3}/2 \, e_4$.

By design $\iota^{-1}$ maps the Eisenstein grid into the standard $\mathbb{Z}^4$ grid. It is not an isometry but it can only change angles, areas and volumes by small constant factors. Thus the widths and sizes of $\iota^{-1}(M)$ and those of its enclosing polytopes have the same asymptotics in $\varepsilon$ as before.

Finally, let $P$ stand for the 4-dimensional polytope containing the scaled-out meniscus $\iota^{-1}(M)$. Because Euclidean 4-volume of $P$ can constrained to be within some known constant factor $F$ of the 4-volume of $\iota^{-1}(M)$, an upper bound for the number of potential $\mathbb{Z}^4$ lattice points in $P$ is roughly $F$ times the upper bound for the number of Eisenstein lattice points in $M$. Therefore if we could enumerate all the $\mathbb{Z}^4$ lattice points in $P$ we could enumerate all the lattice points in $M$ by skipping those points that do not belong to $\iota^{-1}(M)$. Thus \emph{we have reduced the enumeration problem to that of enumerating $\mathbb{Z}^4$ lattice points in a polytope} with known maximum overhead factor. We also reiterate that in case we round out the polytope to have rational vertices,  the bit sizes of the vertices are in $O(\log(1/\varepsilon))$.

There is a wealth of available algorithms and software for counting and identifying lattice points in polytopes, see e.g. \cite{Barvinok94}, \cite{Dyer}, \cite{Lenstra83}. In dimensions 2,3,4 in particular the efficiency and numerical stability of those algorithm is well understood \cite{DeLoera}.
For completeness a top level description of one enumeration algorithm is given in the next subsection.

The uptake of the discussion in the subsection \ref{subsec:enumerating} is the following:

suppose $m$ is the number of integer grid points in the polytope $P$ and $w$ is its largest Euclidean width in any direction;

then the grid points can be enumerated in at most $O(m\,(\log(w))^3)$ calls to a certain point-counting oracle $Em$;

such an oracle can be implemented with complexity polynomial in bit sizes of the polytope vertices.

As we have shown, in our specific case $m$ is in  $\log(1/\varepsilon)^{\nu}$, $w$ is in 

$O(\varepsilon^{-1/4}\, \log(1/\varepsilon)^{\nu/4})$ and the bit sizes of the vertices are logarithmic in $1/\varepsilon$. Thus the overall complexity of the grid point enumeration procedure is polynomial in $\log(1/\varepsilon)$.

\subsection{Enumerating integer grid points in a rational polytope} \label{subsec:enumerating}

Here we describe a quick and dirty divide and conquer strategy for enumerating the grid points based on
the existence of a certain \emph{feasibility oracle}.
%\emph{point counting} algorithms such as TODO:REF:BARVINOK or TODO:REF:DYER.

Suppose we have a Boolean oracle $Em$ such that, given a convex polytope $Q \in \mathbb{R}^m$ returns $Em(Q)=true$ if and only there are no interior integer grid points inside the polytope $Q$. Suppose that $Em$ runs in classical time that is also polynomial in bit size of the polytope description.

As a matter of principle we could take Lenstra's feasibility algorithm from \cite{Lenstra83} which is probably chronologically the first known general algorithm of this type.

%As an overkill we can use a point counting algorithm to implement such an oracle.

Let now $P$  be some rational convex polytope of rank $n$. If $Em(P)=true$, there is nothing to enumerate. Suppose $Em(P)=false$.

Clearly, when $n=1$ the task of enumerating all the grid points in $P$ is trivial and can be done by a straightforward iteration.

 When $n>1$, we select an arbitrary axial direction, e.g. $d=e_n$, and find integer $z \in \mathbb{Z}$ that is the integer point on the
 projection of $P$ onto $d$ closest to the middle of the projection. Cut $P$ by the hyperplane $H_z(d)$ orthogonal to $d$ at level $z$ and set $Q_z(d)=H_z(d) \cap P$. (In corner cases, this intersection might be empty. We allow empty polytopes like that and extend the oracle $Em$ to return $true$ in such cases.)

Let us call $P$ \emph{flat in direction $d$} if its width in the direction $d$ is $\leq 1$. If $P$ is flat in direction $d$ and (a) $Q_z(d)$ is non-empty then all (if any) interior grid points of $P$ are also interior grid points of $Q_z(d)$, otherwise (b) $P$ has no interior grid points.

As a very straightforward version of the interior grid point enumeration we can propose the following recursive

\begin{algorithm}[H]
\caption{$\mbox{IP}(n,P)$}
\label{alg:enumeration}
\algsetup{indent=2em}
\begin{algorithmic}[1]
\REQUIRE{$n>0$, $P$ a convex polytope in $n$ dimensions.}
\IF {$n=1$}
\RETURN {List of points by 1-D iteration.} \label{l1}x
\ENDIF
\IF {$\mbox{Em}(P)$}
\RETURN {Empty list.} \label{l2}
\ENDIF
\STATE {Select axial direction $d=e_k, \,1\leq k \leq n$}\label{l3}
\STATE {Select integer $z$ closest to median of the projection of $P$ onto $d$} \label{l4}
\STATE{$Q_z(d) \gets \{x_k=z\} \cap P$} \label{l5}
\IF{$P$ is flat in the direction $d$}
\RETURN {$\mbox{IP}(n-1,Q_z(d))$} \label{l6}
\ENDIF
\STATE{$(P_z',P_z'') \gets$ segments of $P$ by $\{x_k=z\}$} \label{l7}
\RETURN {$\mbox{IP}(n,P_z') \cup \mbox{IP}(n-1,Q_z(d)) \cup \mbox{IP}(n,P_z'')$} \label{l8}
\end{algorithmic}
\end{algorithm}

\begin{prop}
Suppose $n>1$.
In the context of the above algorithm, let $w$ be the maximum width of the polytope $P$ in any direction and let $m$ be the actual number of interior grid points in $P$. In course of full recursion to completion the procedure $\mbox{IP}(n,P)$ does the $n$-dimensional branching as described in lines \ref{l7}, \ref{l8} of the Algorithm \ref{alg:enumeration}  at most
\begin{equation} \label{eq:bound}
m \, (\lceil \log_2(w) \rceil+1)
\end{equation}
\noindent times.
\end{prop}

\begin{proof}
We are going to prove the proposition for a simplified procedure where we select the one and the same direction $d=e_k$ for cutting the $n$-dimensional polytopes every time.

Let us do induction by $m$.

If $m=0$, the recursion is terminated at node 2) and node 5) is visited zero times.

The induction hypothesis is that the bound (\ref{eq:bound}) is valid for any $m' < m$.

For technical reasons we want to treat separately the case when $1 < w < 2$ and hence $\lceil \log_2(w) \rceil = 1$. It is a simple geometric fact that $P$ may contain one or two $(n-1)$-dimensional slices of the form $Q_z(d)$ and in either case in might take up to two bisections to identify them. The worst-case (tightest) upper bound $2 = \lceil \log_2(w) \rceil + 1$ is required when $P$ contains exactly one interior grid point, but even in this case the number of bisections is equal to the upper bound.

In general case we will perform subordinate induction by $lw = \lceil \log_2(w) \rceil$. Situations where $lw \leq 1$ had been already discussed.

Suppose now that $\lceil \log_2(w) \rceil=v>1$ (which also implies $w>2$).
The subordinate induction hypothesis is that the formula (\ref{eq:bound}) is true for any $m' < m$ and also for the
interior grid point count of $m$ in polytopes whose log width is smaller than $\log_2(w)$ by at least some fixed increment $\Delta$.

 We cut $P$ into polytope sections $P',P''$ and let $w',w''$ be the respective widths of these polytopes in the direction $d$. Clearly, one of these widths is $\leq w/2$ and either of the widths is smaller than $w/2+1/2$. Suppose $m',m''$ are the respective true counts of interior grid points in $P',P''$.
 By design $m'+m'' \leq m$ ax
 nd, in particular $m' \leq m, m'' \leq m$.

 We observe the elementary fact that for any $w>2,\,$  $\log_2(w/2+1/2) < log_2(w)-\Delta$ where $\Delta=2-\log_2(3)$, for example.

 Let us dispose, specifically, of the case when one of the $m',m''$ is zero and the other is equal to $m$. Assume, w.l.o.g. that $m'=0$. Then
 $Em(P')=true$, $\mbox{IP}(n,P')$ recursively visits node 5) zero times and by the induction hypothesis $\mbox{IP}(n,P'')$ recursively visits node 5) at most
 $m \, (\lceil \log_2(w'') \rceil+1) \leq m \, (\lceil \log_2(w) \rceil+1)$ times.

 When neither $m'$ nor $m''$ is zero, then each of them is strictly less than $m$ and, by induction hypothesis $\mbox{IP}(n,P')$ and $\mbox{IP}(n,P'')$ jointly would have visited node 5) at most
 $m' (\lceil \log_2(w') \rceil+1) + m'' (\lceil \log_2(w'') \rceil+1)$ times by the completion of the recursion.

 From the above discussion it follows that each of the numbers $\lceil \log_2(w') \rceil, \, $ $\lceil \log_2(w'')\rceil$ is less or equal than $\lceil \log_2(w) \rceil$ therefore
 $m' (\lceil \log_2(w') \rceil+1) + $ 
 
 $m'' (\lceil \log_2(w'') \rceil+1) \leq$ $(m'+m'') \, (\lceil \log_2(w) \rceil +1)
 \leq m\, (\lceil \log_2(w) \rceil+1)$.

\end{proof}

Thus we have bounded the overall number of bisections by (\ref{eq:bound}) and the actual interior grid points are located on $(n-1)$-dimensional slices of the $Q_z(d)$ type. Given $n > 2$, to make the upper bound we assume that recursion into each $Q_z(d)$
 would requires essentially the same number of bisections (in $(n-1)$ dimensions).

Assuming that the enumeration of the grid points in a one-dimensional polytope has trivial cost, we thus conclude that the overall number of non-trivial operations and thus the number of calls to the oracle $Em$ performed by $\mbox{IP}(n,P)$ over all levels of recursion can be upper-bounded by $O(m\,(\log(w))^{n-1})$. This bound is a fairly ad hoc one, but it fits the polynomiality claim we have been looking for.

In dimension up to $4$ (which is the scope of its application to the metaplectic state synthesis) the algorithm is likely to be sufficiently performant in practice.

To make an enumeration algorithm practically more appealing, we should move away from using the feasibility oracle $\mbox{Em}$ as a black box. Consider constructive steps involved in building Lenstra's feasibility procedure in \cite{Lenstra83}.
Suppose, for simplicity, that $P$ is a generic simplex, i.e. the convex hull of $(n+1)$ points in general position. There is a simple closed-form linear transformation $\tau$ that transforms $P$ into the standard simplex spanned by the standard Euclidean basis $e_1, \ldots, e_n$.

Lenstra observes that, instead of characterizing integer grid point inside $P$ it suffices to characterize the intersection of the lattice $L=\mathcal{L}(b_1=\tau(e_1), \ldots, b_n=\tau(e_n))$ with the $\tau(P)$ which is the standard simplex with known Euclidean geometry. There is an effective constant $K$ depending only on the latter geometry such that in at most $K$ steps of known unit cost we either (a) find at least one lattice point $y \in L \cap \tau(P)$ or, (b) we prove that the intersection is empty.

Suppose it is (a). By the time we have found a  $y \in L \cap \tau(P)$ we would have LLL-reduced the basis $b_1, \ldots, b_n$ to a nicer, ``more orthogonal'' basis $b'_1, \ldots, b'_n$ worthwhile holding on to. We can recursively search for more points of $L \cap \tau(P)$ on some hyperplane $H$ passing through $y$ and orthogonal to, say, $b_n$. After saving out these points, we can bisect the $\tau(P)$ and recursively run the algorithm on the sections. Note that we never do such a bisection in any dimensionality unless we have just saved out at least one of the desired grid points. Therefore the total number of bisections after ultimate completion of the recursion does not exceed the original number of the grid points in $P$. This is likely to lead to a substantial practical improvement over the algorithm \ref{alg:enumeration}.

\section{Important special case: the Clifford+$P_9$ basis} \label{sec:P9}

Research in \cite{BCRS} emphasized the benefits of a certain universal quantum basis formed by the ternary Clifford group and the diagonal gate
$P_9 = e^{-2 \pi i/9} |0\rangle \langle 0| + |1\rangle \langle 1| + e^{2 \pi i/9} |2\rangle \langle 2|$.

The basis provides a more natural language for describing important reversible classical circuits such as natural and modular integer adders, multipliers and modular exponentiation. However the $P_9$ gate is not necessarily the best for synthesis of non-classical gates and also the gate itself requires rather costly magic state distillation.

In this section we show that the magic state $\mu$ required for the $P_9$ can be emulated to required fidelity very efficiently in Clifford+$R$ framework. We do it this way since it allows to run the $P_9$ gate in constant online width while having the offline magic state preparation of moderate width that is asymptotically smaller than the width of known magic state distillation protocols.

Recall (\cite{CampbellEtAl}) that the $P_9$ can be performed by a deterministic constant depth state injection protocol where the magic state injected is
\[
\mu = (e^{-2 \pi i/9} |0\rangle  + |1\rangle  + e^{2 \pi i/9} |2\rangle)/\sqrt{3}
\]

Instead of distilling $\mu$ with some distillation protocol, $\mu$ can be prepared offline to a certain precision as
$\tilde{P_9}\,H\,|0\rangle$ where $\tilde{P_9}$ is an offline approximation of the $P_9$. A convenient representation of $P_9$ is $\tau_{0,2} R_{\phi}$ where $\tau_{0,2}$ is a Clifford transposition and $R_{\phi}$ is the Householder reflection about the two-level vector $\phi = (-e^{-\pi i/9} |0\rangle  + e^{\pi i/9} |2\rangle)/\sqrt{2}$. Suppose $c_{\phi}$ is a metaplectic circuit that prepares $\phi$ from $|2\rangle$. Then $R_{\phi} = c_{\phi}\,R\,c_{\phi}^{-1}$ and the $R$-count of the latter circuit is less or equal than $2\, \mbox{R-count}(c_{\phi}) +1$.

In this design line it is paramount to optimally approximate $\phi$ by a metaplectic state to an arbitrary desired precision. It follows from the analysis below that $\phi$ is not a strongly exceptional state. However, it turns out to be to be somewhat exceptional: a theoretical bound established in Proposition \ref{prop:phi;approx} claims that lattice approximations of $\phi$ are $1/2 \, \log_3(1/\varepsilon)$ more costly than the uniform asymptotic lower bound on the $R$-count; numerical simulation suggests that available lattice approximations of $\phi$ specifically are in line with that cost estimate.

\begin{prop} \label{prop:phi;approx}
For arbitrarily small $\varepsilon$ a certain $\varepsilon$-approximation of the state $\phi = (-e^{-\pi i/9} |0\rangle  + e^{\pi i/9} |2\rangle)/\sqrt{2}$ can be prepared by a metaplectic circuit with the $R$-count at most $3 \, \log_3(1/\varepsilon)+ O(\log \log 1/\varepsilon)$.

Such a circuit of optimal overall depth can be synthesized classically in classical time that is polynomial in
$\log(1/\varepsilon)$.
\end{prop}

\begin{proof}
As a matter of proof, we will reduce the approximation task to one in a certain two-dimensional lattice with fully reduced orthogonal basis.

The $\mathbb{R}^4$ representation of the vector $\phi$ is
\[
p=(-\cos(\pi/9), \sin(\pi/9), \cos(\pi/9), \sin(\pi/9))/\sqrt{2}
\]
\noindent This vector not in general position with respect to the Eisenstein lattice: it is exactly orthogonal to the rank two sublattice generated by the following two vectors
\[
d_1 = v_1+v_3; d_2 = -v_1-2\,v_2 + v_3 + 2 \, v_4
\]
\noindent where $v_1,v_2,v_3,v_4$ are as defined by the equations (\ref{eq:v:basis}),(\ref{eq:v:basis:contd}).

Let $P$ be the two-dimensional Euclidean orthogonal complement to $\mbox{span}(d_1,d_2)$. Clearly, $p \in P$.
If $y$ is an Eisenstein lattice vector that approximates some $\sqrt{3}^k \,p, \, k \in \mathbb{Z}$ then the approximation precision depends only on the projection of $y$ onto $P$ in the first place.

It is important to note upfront, right here, that in case $y$ is not $k$-feasible, we can attempt making it $k$-feasible by adding some $k_1 \, d_1 + k_2 \, d_2, \, k_1,k_2 \in \mathbb{Z}$ to it. This does not affect the precision at all, but should (and does) move us to a $k$-feasible candidate from a possibly infeasible $y$. Assuming Conjecture \ref{cnj:norm:eq} such $k$-feasible $y+k_1 \, d_1 + k_2 \, d_2$ is found with probability arbitrarily close to 1 upon inspecting some $O(k)$ pairs $(k_1,k_2)$.

Assuming the feasibility is taken care of, the geometric part of the approximation problem is reduced the classical Closest Vector Problem for the target vector $\sqrt{3}^k \,p$ in the two-dimensional lattice that is the projection of the Eisenstein lattice onto the two-dimensional plane $P$, since, instead of enumerating the lattice points in $P$ we need to find just one point that is sufficiently close for the given $k$.

It is easy to observe that the vectors $p_1=(-1,0,1,0)/\sqrt{2},\, p_2=(0,1,0,1)/\sqrt{2}$ form a Euclidean orthonormal basis in $P$. By direct computation it is equally straightforward to establish that the orthogonal projection of the Eisenstein lattice $\mathcal{E}$ onto $P$ is the rank two lattice generated by
$b_1 = 1/\sqrt{2}\, p_1$ and $b_2 = \sqrt{3/2}\, p_2$ where the two generating vectors form a fully reduced orthogonal lattice basis. The optimal enumeration of such a simple lattice in an $\varepsilon$-neighborhood of a target vector is well understood since \cite{Ross},\cite{BBG}. When the goal is to $\varepsilon$-approximate a unit vector by a vector of the form $(a_1 \, b_1 + a_2 \, b_2)/\sqrt{3}^k, \, a_1,a_2,k \in \mathbb{Z}$ then the uniform asymptotic lower bound on $k$ is $3 \log_3(1/\varepsilon)$ and for non-exceptional target vectors this bound can be also met up to a residual $O(\log \log 1/\varepsilon)$ term. In terms of the two-dimensional basis $(b_1,b_2)$ our target vector
$\phi=\sqrt{2} ( \cos(\pi/9)\, b_1 + 1/\sqrt{3} \, \sin(\pi/9) \, b_2$ and non-exceptional in this context.

\end{proof}

\begin{observ} [Based on numerical simulation]
Numerical experiments based on ad hoc two-dimensional Closest Vector method suggest that, in the context of the Proposition \ref{prop:phi;approx} an $\varepsilon$-approximation of the state $\phi = (-e^{-\pi i/9} |0\rangle  + e^{\pi i/9} |2\rangle)/\sqrt{2}$ can be prepared by a metaplectic circuit with the $R$-count at most

\noindent ${\bf 3} \, \log_3(1/\varepsilon)+ O(\log \log 1/\varepsilon)$
\end{observ}

Examples of lattice points found in this experiment are presented in the Appendix \ref{appendix:a}.

\section{Concluding remarks}

We have shown that a single-qutrit two-level unitary state $|\psi\rangle$ in general position can be $\varepsilon$-approximated effectively by a metaplectic state of the form (\ref{eq:exact:state}) with $k$ in at most $5/2 \, \log_3(1/\varepsilon) + O(\log \log 1/\varepsilon)$. If a \emph{norm equation oracle} is available then an approximating state with truly minimal $k$ can be found. However even when the true minimal $k_{min}$ is known, the depth of the approximating circuit can take up one of the three values $\{2\,k_{min}-1,2\,k_{min},2\,k_{min}+1\}$. Therefore we will have to iterate through all the candidates with $k=k_{min}$ in order to determine, which of the three values is a circuit depth minimum.
Note that candidate states (\ref{eq:exact:state}) with $k>k_{min}$ are not going to generate approximation circuits with depth smaller than $2\,k_{min}+1$.

The problem of \emph{truly minimal $R$-count} is less well defined. While in the above context the $R$-count of a metaplectic circuit is upper-bounded by $k_{min}+1$, it can be in principle significantly smaller. In other words, the $R$-count is not a well-defined function of depth and a metaplectic circuit proxy of larger depth might have a smaller $R$-count.

Our specific way of synthesizing a circuit for the Householder reflection $R_{|\psi\rangle}$ is by reducing it to $R_{|2\rangle}$ which implies the existence of a circuit of depth no more than $2\,\ell+1$ for the reflection, where $\ell$ is the minimal depth of a metaplectic circuit to prepare $|\psi\rangle$.  We make no representation here that some of the reflections cannot be obtained by some completely different scheme at a smaller depth. (This issue remains open for future research.)

It is worthwhile to note that all the designs of section \ref{sec:main:lemma} can be applied with much more ease to circuit synthesis over the single-qubit $V$ basis (cf. \cite{BGS}). Ostensibly they could be more costly in practice than the custom algorithms from \cite{BBG} or \cite{Ross}, but this might be a price worth paying for conceptual transparency.

\section{Acknowledgements}

The Author wishes to thank V. Kluichnikov for valuable literature references: especially \cite{LLL},\cite{Sardari},  and for pointing out algorithmic versions of Khinchin's flatness theorem that provide for enumeration of lattice points in polynomial time. The Author also wishes to dedicate this research to his father, Vadim Bocharov.

\appendix
\section{Examples of approximation vectors for the magic state generator $|\phi\rangle$} \label{appendix:a}

Recall that the state $\phi = (-e^{-\pi i/9} |0\rangle  + e^{\pi i/9} |2\rangle)/\sqrt{2}$ introduced in section \ref{sec:P9} is used for emulating the important $P_9$ gate that has been extensively used for circuit synthesis in \cite{BCRS}.
If $\phi$ is prepared to some precision by a circuit of depth $\ell$ (resp. $R$-count $r$), then $P_9$ can be emulated to essentially the same precision by a corresponding circuit of depth at most $2\, \ell +1$ (resp. $R$-count at most $2\, r+1$). If the emulation is set up offline then using the $P_9$ gate available at certain fidelity one can prepare the ``$P_9$'' magic state $\mu$ offline at the same fidelity.
%$k=30$ $\{19068804, 4006785\}$ $-9.53$ $1.41$

%$k=60$ $\{273614088661424, 57496802915208\}$ $-19.973$ $0.08$

%$k=90$ $\{3926063111969316475316, 825016278023092749328\}$ $-29.527$ $1.42$

%$=\mbox{dist}(s/3^{k/2},|\phi\rangle)$

Table \ref{table:examples:proxy} below offers three examples of resource states for circutizing the state $|\phi\rangle$ to the best possible precision. Taking $(s,k)$ pair from each row one approximates $|\phi\rangle$ to precision $\varepsilon$ shown in the third column by $s/\sqrt{3}^k$. One adjusts $s$ with a small shift in order to get a $k$-feasible state $s'$ (definition \ref{def:k:feasible}) as shown in the last column. This adjustment does not change the precision. As per the definition, $s'/\sqrt{3}^k$ can be represented by a metaplectic circuit with $R$-count at most $k+1$ and overall depth between $2 k-1$ and $2 k +1$.

\begin{table}
    \begin{tabular}{c|c|c|c|c}
         $k$  &  \multicolumn{1}{c|}{State} &  \multicolumn{1}{c|}{$\varepsilon$} & \multicolumn{1}{c|} {$k+3\log_3(\varepsilon)$} & $k$-feasible state\\
    \hline
    30        & $s=(-7531010+ 4006784 \omega_3)|0\rangle + $                         & $3^{-9.53}$ & $1.41$ & $s'=s+$  \\
            & $(11537794+4006784\omega_3)|2\rangle$ & & & $4 (|0\rangle+|2\rangle)$\\
    60        & $s=(-108058642873108+ 57496802915208 \omega_3)|0\rangle + $  & $3^{-19.973}$  &$0.08$ & $s'=s+$  \\
            &$(165555445788316+ 57496802915208\omega_3)|2\rangle$ & & & $3 (|0\rangle+|2\rangle)$\\
    90        & $s=(-1550523416973111862994 +$   & $3^{-29.527}$ & $1.42$& $s'=s+$\\
            & $825016278023092749328 \omega_3)|0\rangle + $ & & & $27 (|0\rangle+|2\rangle)$ \\
            & $(2375539694996204612322+$ &  & & \\
            & $825016278023092749328\omega_3)|2\rangle$ & & &\\

    \end{tabular}
    \caption{Examples of proxy states for $|\phi\rangle$. The first column shows the scaling level $k$, the second column contains a closest unnormalized Eisenstein lattice state. The third column contains precision achievable at the selected scaling level: $\varepsilon =\mbox{dist}(s/3^{k/2},|\phi\rangle)$. The fourth column measures the residual terms between $k$ and $3\, \log_3(1/\varepsilon)$.The last column proposes a small shift that turns the proxy state $s$ into a $k$-feasible state.} \label{table:examples:proxy}
  \end{table}

\end{document}